\documentclass[preprint,12pt,authoryear]{elsarticle}

\usepackage{amssymb}
\usepackage{graphicx}
\usepackage{braket}
\usepackage{amsmath}
\usepackage{hyperref}
\usepackage{float}
\usepackage{amsthm}
\usepackage{longtable}
\usepackage{booktabs}
\usepackage{CJKutf8}
\usepackage[strings]{underscore}

\newtheorem{theorem}{Theorem}
\begin{document}
\begin{CJK}{UTF8}{gbsn}

\begin{frontmatter}

\title{QUBO Formulations for Variation of Domination Problem}

\author[a]{Haoqian Pan}
\author[a]{Changhong Lu}
\affiliation[a]{organization={School of Mathematical Sciences,  Key Laboratory of MEA(Ministry of Education) \& Shanghai Key Laboratory of PMMP, East China Normal University}, 
            city={Shanghai},
            postcode={200241},
            country={China}}

\begin{abstract}

With the development of quantum computing, the use of quantum algorithms to solve combinatorial optimization problems on quantum computers has become a major research focus. The Quadratic Unconstrained Binary Optimization (QUBO) model serves as a bridge between combinatorial optimization problems and quantum computers, and is a prerequisite for these studies. In combinatorial optimization problems, the Domination Problem (DP) is related to many practical issues in the real world, such as the fire station problem, social network theory, and so on. Additionally, the DP has numerous variants, such as independent DP, total DP, k-domination, and so forth. However, there is a scarcity of quantum computing research on these variant problems. A possible reason for this is the lack of research on QUBO modeling for these issues. This paper investigates the QUBO modeling methods for the classic DP and its variants. Compared to previous studies, the QUBO modeling method we propose for the classic DP can utilize fewer qubits. This will lower the barrier for solving DP on quantum computers. At the same time, for many variants of DP problems, we provide their QUBO modeling methods for the first time. Our work will accelerate the entry of DP into the quantum era.
\end{abstract}


\end{frontmatter}


\section{Introduction}

In recent years, with the progress of quantum computing \citep{RN432,RN394,RN431,RN437,RN438}, and quantum algorithms such as Quantum Annealing (QA) \citep{RN433}, Quantum Approximate Optimization Algorithm (QAOA) \citep{RN343}, Variational Quantum Eigensolver (VQE) \citep{RN434}, and so forth, an increasing number of combinatorial optimization problems are being attempted to be solved using quantum computers. For instance, max-cut \citep{RN343}, minimum vertex cover \citep{RN334}, perfect matching \citep{RN435}, maximum clique \citep{RN425}, and so on. The key to solving combinatorial optimization problems with quantum computers is to convert the problem into a Quadratic Unconstrained Binary Optimization (QUBO) problem. The basic QUBO formulation is given by Eq. \ref{eq:qubo} where $x$ is the vector of binary variables and $Q$ is the matrix of constants known as the QUBO matrix. 

\begin{equation}
    minimize/maximize \quad y = x^{t}Qx \label{eq:qubo}
\end{equation}

The combinatorial optimization problem discussed in this paper is the Domination Problem (DP), which is a classic issue in graph theory. A dominating set \( D \) of \( G = (V, E) \) is the subset of \( V \), and for every \( u \in V \setminus D \), there is at least one node in \( D \) which is adjacent to \( u \) as shown in Fig. \ref{fig:dsp}. The DP is about finding the smallest such \( D \). The DP is related to many practical issues in the real world, such as the fire station problem, sets of representatives, social networks, and so on. Besides the classic problem, there are many variants of the DP problem, such as independent DP, total DP \citep{RN436}, and so forth.

\begin{figure}[H]
    \centering
    \includegraphics[width=10cm]{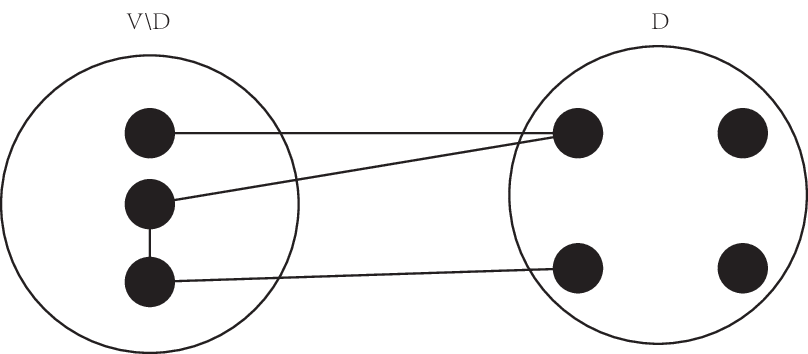}
    \caption {The illustration of DP}
    \label{fig:dsp}
\end{figure}

Utilizing quantum computers to solve DP, especially its variant problems, is an area where the literature is almost non-existent. Currently, there are only a few studies on solving DP using quantum computers. For instance, \cite{RN416} was the first to use the IBM Poughkeepsie machine to solve the classical DP. We believe this scarcity is related to the lack of research on QUBO modeling for the DP and its variants, as QUBO modeling is the bridge connecting combinatorial optimization and quantum computing. This paper primarily investigates the QUBO modeling of DP and its variants. We first model both the DP and its variants as 0-1 integer programming models. Then, we organize and summarize the constraints of these problems and provide methods for converting them into quadratic penalties. By adding quadratic penalties to the original objective function, we obtain the corresponding QUBO model. Our main contributions are: (1) Compared to \cite{RN416}, which requires \( |V|^2 \) variables to model the DP, we propose a new method for modeling the DP as a QUBO problem that only requires \( |V| + 2|E| \) variables. (2) We have performed QUBO modeling for most of the variant problems compiled in \cite{RN436}. Our research builds a bridge between the DP and its variants and quantum computing. This will accelerate the solution of DP and its variant problems into the quantum era.

The subsequent organization of this paper is as follows. In Section \ref{sec:review}, we review the literature related to QUBO modeling or quantum computing. In Section \ref{sec:method}, we introduce the method for transforming the DP and its variants into QUBO models, and provide their respective models in conjunction with a specific graph \( G \). In Section \ref{sec:conclusion}, we summarize the entire paper.

\section{Literature review}\label{sec:review}

In recent years, with the advancement of quantum computing, an increasing number of studies have attempted to transform combinatorial optimization problems into QUBO models. A significant portion of these studies focuses on classic problems in graph theory. \cite{RN417} discussed the QUBO formulation of many combinatorial optimization problems in graph theory, such as partitioning problems, covering and packing problems, coloring problems, and Hamiltonian cycles. \cite{RN418} demonstrated and compared various QUBO modeling methods for graph isomorphism problems, and their proposed direct formulation method has been experimentally verified to be more feasible. \cite{RN425} modeled the maximum clique problem as a QUBO model and solved this problem on the D-Wave 2X quantum annealer. \cite{RN426} solved the maximum flow problem in QUBO form using the D-Wave quantum annealing machine with the DW_2000Q_6 chip under different cases. Similar studies include weighted maximum cut and weighted maximum clique \citep{RN427}. Additionally, \cite{RN414} provided a more systematic introduction to converting combinatorial optimization models into QUBO models from a mathematical perspective. Unlike the classic problems in graph theory, some literature considered converting traditional machine learning field problems into QUBO models. \cite{RN423} gave the formulation of the k-medoids problem, such as the $L_{1}$-norm given by \cite{RN422}. \cite{RN424} considered linear regression, support vector machines, and k-means clustering. \cite{RN429} looked at balanced k-means clustering. In addition, there are many documents that transform practical problems in the industry into QUBO problems for solving. \cite{RN428} transformed the satellite scheduling problem into a QUBO model. \cite{RN430} considered a practical railway dispatching problem from the Polish railway network and solved the QUBO problems using D-Wave quantum annealers. \cite{RN419} used the quantum annealing algorithm to solve the vehicle routing problem, while \cite{RN420} and \cite{RN421} further considered the capacitated vehicle routing problem and the multi-Depot capacitated vehicle routing problem. It can be seen that with the development of quantum computing, traditional combinatorial optimization fields, machine learning fields, and practical problems in the industry are all being studied on how to model them as QUBO problems to try to solve them on quantum computers.

\section{Method} \label{sec:method}

This section will provide a detailed introduction on how to convert the DP and its variants into QUBO models. We begin by formulating the DP and its variants into 0-1 integer programming models. We then transform the constraints in these models into quadratic penalties. Finally, by adding these quadratic penalties to the original objective function, we obtain the QUBO models for these problems. In Section \ref{sec:method1}, we provide the specific definitions of the DP and its variants. In Section \ref{sec:method2}, we present the 0-1 integer programming models for the DP and seven of its variants. Following that, in Section \ref{sec:method3}, we categorize the constraint conditions of these models into four types and provide methods for converting them into quadratic penalties. Lastly, in Section \ref{sec:method4}, we combine specific examples to show how, by adding quadratic penalties to the original objective function, we derive the QUBO models for these problems.

\subsection{Definitions of DP and its variation}\label{sec:method1}

We begin by providing the definition of the classic DP. Given a graph \( G = (V, E) \), a dominating set \( D \) of graph \( G \) is defined as a subset of \( V \), such that for any vertex \( v \) in \( V \setminus D \), there is at least one neighbor in \( D \). The number of vertices in such a \( D \) is referred to as the domination number \( \gamma \). The DP is to find the smallest such \( \gamma \). In addition to the classic DP, there is a multitude of variant problems. For the convenience of the reader, we have compiled these problems in Table. \ref{table:defdsps}. In the table, \( G[D] \) is the subgraph of graph \( G \) induced by \( D \), \( N[v] \) and the subsequently used \( N(v) \) are the closed and open neighborhoods of \( v \), respectively. \( d(u, v) \) is the shortest distance between \( u \) and \( v \).

\begin{longtable}{|c|c|c|}
  \caption{Variation of domination set problems}
  \label{table:defdsps}\\
  \hline
  type of DP & label & definition\\
  \hline
  \endfirsthead
  \multicolumn{3}{r}{Continued}\\
  \hline
  type of DP & label & definition\\
  \hline
  \endhead
  \hline
  \multicolumn{3}{r}{Continued on next page}\\
  \endfoot
  \endlastfoot
  \hline
  classical&$\gamma$& \\
  \hline
  independent&$\gamma_{i}$&$G[D]$ has no edges.\\
  \hline
  total &$\gamma_{t}$&$G[D]$ has no isolated vertices.\\
  \hline
  perfect &$\gamma_{per}$&   $|N[v] \cap D| = 1$ for any $v \in V \backslash D$.\\
  \hline
  clique &$\gamma_{cl}$& $D$ is a clique.\\
  \hline
  independent and perfect &$\gamma_{iper}$& $D$ is independent and perfect.\\
  \hline
  total and perfect &$\gamma_{tper}$& $D$ is total and perfect.\\
  \hline
  k-domination &$\gamma_{k}$& $\forall v \in V, \exists u \in D$, $d(u,v) \leq k$\\
  \hline
\end{longtable}

\subsection{Mathematical model}\label{sec:method2}

We model the classic problem and its variants as 0-1 integer programming models. we define a binary variable \( x_v \) for each vertex \( v \in V \), where: \( x_v = 1 \) if vertex \( v \) is included in the dominating set \( D \), \( x_v = 0 \) otherwise.

\begin{itemize}
  \item $\gamma$
  \begin{alignat}{2}
    \min_{\{X_{i}\}} \quad & \sum\limits_{i=1}^{|V|} X_{i} & \tag{$\gamma$.1}\\
    \mbox{s.t.}\quad
    &\sum\limits_{j \in N[i]} X_{j} \ge 1 \quad \forall i \in V  & \tag{$\gamma$.2} \label{const:c2}\\
    &X_{i} \in \{0,1\}  \quad \forall i \in V &\tag{$\gamma$.3}
  \end{alignat}
  
  Constraint \ref{const:c2} ensures that at least one vertex in the closed neighborhood of each vertex belongs to the set \( D \).

  \item $\gamma_{i}$
  \begin{alignat}{2}
    \min_{\{X_{i}\}} \quad & \sum\limits_{i=1}^{|V|} X_{i} & \tag{$\gamma_{i}$.1}\\
    \mbox{s.t.}\quad
    &\sum\limits_{j \in N[i]} X_{j} \ge 1 \quad \forall i \in V  & \tag{$\gamma_{i}$.2}\\
    &X_{i} \in \{0,1\}  \quad \forall i \in V &\tag{$\gamma_{i}$.3}\\
    &X_{i} * X_{j} = 0 \quad \forall ij \in E &\tag{$\gamma_{i}$.4} \label{const:i4}
  \end{alignat}

  Constraint \ref{const:i4} ensures that the dominating set \( D \) is independent.

  \item $\gamma_{t}$
  \begin{alignat}{2}
    \min_{\{X_{i}\}} \quad & \sum\limits_{i=1}^{|V|} X_{i} & \tag{$\gamma_{t}$.1}\\
    \mbox{s.t.}\quad
    &\sum\limits_{j \in N(i)} X_{j} \ge 1 \quad \forall i \in V  & \tag{$\gamma_{t}$.2}\\
    &X_{i} \in \{0,1\}  \quad \forall i \in V &\tag{$\gamma_{t}$.3}
  \end{alignat}
  \item $\gamma_{per}$
  \begin{alignat}{2}
    \min_{\{X_{i}\}} \quad & \sum\limits_{i=1}^{|V|} X_{i} & \tag{$\gamma_{per}$.1}\\
    \mbox{s.t.}\quad
    &\sum\limits_{j \in N[i]} X_{j} \ge 1 \quad \forall i \in V  & \tag{$\gamma_{per}$.2}\\
    &X_{i} \in \{0,1\}  \quad \forall i \in V &\tag{$\gamma_{per}$.3} \\
    &\sum\limits_{ij \in E} [X_{i} * (1 - X_{j}) + X_{j} * (1 - X_{i}) ]  = |V| - \sum\limits_{i}^{|V|} X_{i} &\tag{$\gamma_{per}$.4}
  \end{alignat}

  $|V| - \sum\limits_{i}^{|V|} X_{i}$ denotes the number of vertexes in $V$\textbackslash$D$ and $\sum\limits_{ij \in E} [X_{i} * (1 - X_{j}) + X_{j} * (1 - X_{i}) ]$ can count the number of edges from $V$\textbackslash$D$ to $D$.

  \item $\gamma_{cl}$
  \begin{alignat}{2}
    \min_{\{X_{i}\}} \quad & c = \sum\limits_{i=1}^{|V|} X_{i} & \tag{$\gamma_{cl}$.1}\\
    \mbox{s.t.}\quad
    &\sum\limits_{j \in N[i]} X_{j} \ge 1 \quad \forall i \in V  & \tag{$\gamma_{cl}$.2}\\
    &X_{i} \in \{0,1\}  \quad \forall i \in V &\tag{$\gamma_{cl}$.3} \\
    & \frac{1}{2}c(c - 1) =  \sum\limits_{ij \in E} X_{i} * X_{j} &\tag{$\gamma_{cl}$.4}
  \end{alignat}
  \item $\gamma_{iper}$
  \begin{alignat}{2}
    \min_{\{X_{i}\}} \quad & \sum\limits_{i=1}^{|V|} X_{i} & \tag{$\gamma_{iper}$.1}\\
    \mbox{s.t.}\quad
    &\sum\limits_{j \in N[i]} X_{j} \ge 1 \quad \forall i \in V  & \tag{$\gamma_{iper}$.2}\\
    &X_{i} \in \{0,1\}  \quad \forall i \in V &\tag{$\gamma_{iper}$.3}\\
    &X_{i} * X_{j} = 0 \quad \forall ij \in E &\tag{$\gamma_{iper}$.4} \\
    &\sum\limits_{ij \in E} [X_{i} * (1 - X_{j}) + X_{j} * (1 - X_{i}) ]  = |V| - \sum\limits_{i}^{|V|} X_{i} &\tag{$\gamma_{iper}$.5}
  \end{alignat}
  \item $\gamma_{tper}$
  \begin{alignat}{2}
    \min_{\{X_{i}\}} \quad & \sum\limits_{i=1}^{|V|} X_{i} & \tag{$\gamma_{tper}$.1}\\
    \mbox{s.t.}\quad
    &\sum\limits_{j \in N(i)} X_{j} \ge 1 \quad \forall i \in V  & \tag{$\gamma_{tper}$.2}\\
    &X_{i} \in \{0,1\}  \quad \forall i \in V &\tag{$\gamma_{tper}$.3}\\
    &\sum\limits_{ij \in E} [X_{i} * (1 - X_{j}) + X_{j} * (1 - X_{i}) ]  = |V| - \sum\limits_{i}^{|V|} X_{i} &\tag{$\gamma_{tper}$.4}
  \end{alignat}
  \item $\gamma_{k}$
  \begin{alignat}{2}
    \min_{\{X_{i}\}} \quad & \sum\limits_{i=1}^{|V|} X_{i} & \tag{$\gamma_{k}$.1}\\
    \mbox{s.t.}\quad
    &\sum\limits_{j \in N^{\prime}[i]} X_{j} \ge 1 \quad \forall i \in V  & \tag{$\gamma_{k}$.2}\\
    &X_{i} \in \{0,1\}  \quad \forall i \in V &\tag{$\gamma_{k}$.3}
  \end{alignat}

  Our approach to constructing the model for the \( k \)-domination problem involves solving the classical DP in a new graph \( G' \). The graph \( G' \) is constructed by adding edges between vertices \( u \) and \( v \) for which \( uv \notin E \) and the distance \( d(u,v) \leq k \). Taking Fig. \ref{fig:kdom} as an example, assume that each edge in the graph has a length of 1. When we need to solve the \( 2 \)-domination problem, we connect vertex pairs \( u, v \) in the original graph \( G \) that satisfy the conditions \( uv \notin E \) and \( d(u,v) \leq 2 \), thereby forming \( G' \). We then solve the classical DP within \( G' \).

  \begin{figure}[H]
    \centering
    \includegraphics[width=10cm]{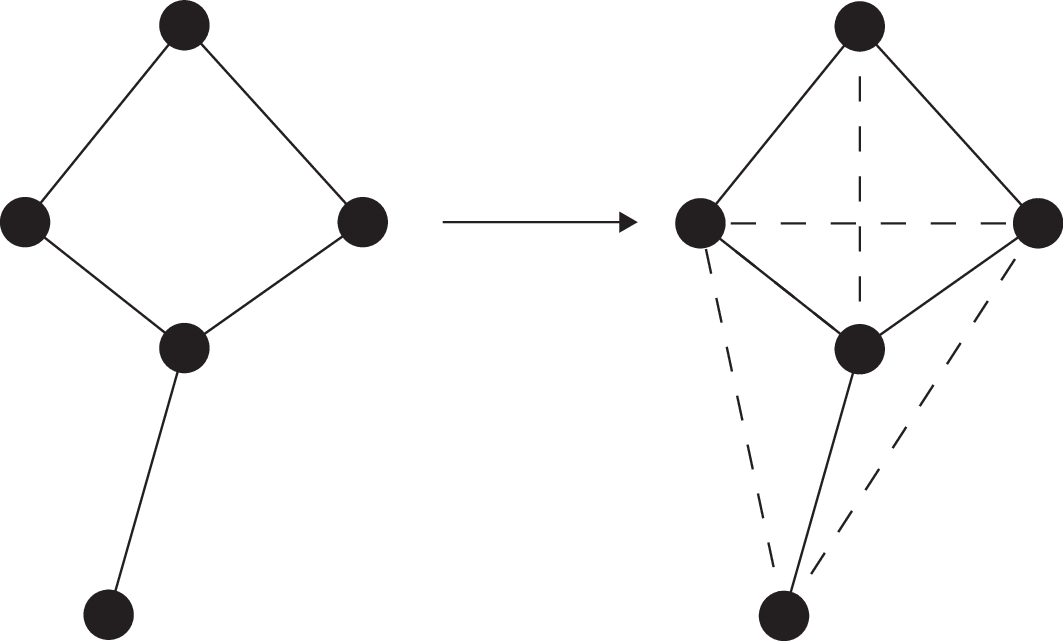}
    \caption {The illustration of building $G^{\prime}$ where $k = 2$.}
    \label{fig:kdom}
  \end{figure}
\end{itemize}

\subsection{Quadratic penalties}\label{sec:method3}

By analyzing the constraints of the variation of DP, We can observe that the constraints of all types of problems are composed of the following basic constraints. And structurally speaking, Constraints \ref{eq:basic}, \ref{eq:total} and \ref{eq:kdom} are the same. 

\begin{align}
  &\sum\limits_{j \in N[i]} X_{j} \ge 1 \quad \forall i \in V \tag{$\alpha$} \label{eq:basic}\\
  &X_{i} * X_{j} = 0 \quad \forall ij \in E \tag{$\beta$} \label{eq:independent}\\
  &\sum\limits_{j \in N(i)} X_{j} \ge 1 \quad \forall i \in V \tag{$\delta$} \label{eq:total}\\
  &\sum\limits_{ij \in E} [X_{i} * (1 - X_{j}) + X_{j} * (1 - X_{i}) ]  = |V| - \sum\limits_{i}^{|V|} X_{i} \tag{$\epsilon$} \label{eq:perfect} \\
  &\frac{1}{2}c(c - 1) =  \sum\limits_{ij \in E} X_{i} * X_{j} \tag{$\zeta$} \label{eq:clique} \\
  &\sum\limits_{j \in N^{\prime}[i]} X_{j} \ge 1 \quad \forall i \in V \tag{$\eta$} \label{eq:kdom}
\end{align}
A classical way to process the constraint during QUBO formulation is to add a penalty to the raw objective function. And a big challenge is to convert the constraints into quadratic penalties. Therefore, this will be the focus of discussion in this chapter. Next, we will categorize these constraints and present methods for converting them into quadratic penalties. \newline
\begin{itemize}
\item Constraints \ref{eq:basic}, \ref{eq:total}, \ref{eq:kdom}\\
    
For constraints \ref{eq:basic}, \ref{eq:total}, and \ref{eq:kdom}, they can be generally described in the following form:
\begin{equation}
X_{1} + X_{2} + \dots + X_{n} \geq 1 \label{eq:c1normal}  
\end{equation}

When \( n = 1 \), this condition can be transformed into \( P(X_{1} - 1)^2 \). Here, \( P \) is the penalty coefficient. When \( n = 2 \), according to \cite{RN414}, it can be converted to \( P(1 - X_{1} - X_{2} + X_{1} * X_{2}) \). Therefore, we will mainly deal with the case where \( n \geq 3 \). For this situation, the literature \cite{RN414} points out that slack variables need to be introduced to convert inequality constraints into equality constraints, and the slack variables should be represented in the form of a combination of 0-1 variables. However, it does not provide a standardized form. To further illustrate this issue, we need to define the notation \( bl_{n} \) as the length of the binary number of \( n \), and similarly, \( bl_{n-1} \) is the binary length of \( n - 1 \). After introducing the slack variables, constraint \ref{eq:c1normal} will become:
\begin{equation}
  X_{1} + X_{2} + \dots + X_{n}  - S = 1 \label{eq:scidea}
\end{equation}

It can be determined that \(\{S\} = \{0, 1, 2, \ldots, n-1\}\), and then, we can express \( S \) in the form of a combination of 0-1 variables:
\begin{equation}
    S = C_{1}X_{1}^{\prime} + C_{2} * X_{2}^{\prime} + C_{3} * X_{3}^{\prime} + \dots \label{eq:c1sc}
\end{equation}
Where \( X_{*}^{\prime} \in \{0,1\} \), \( C_{*} \) are coefficients. And the range of \( S \) represented by Eq. \ref{eq:c1sc} must include the estimate of \( S \) in Eq. \ref{eq:scidea}, that is, \(\{0, 1, 2, \dots, n-1\}\). Once we construct such an \( S \), we can transform constraint \ref{eq:c1normal} into:
\begin{equation}
    P * [X_{1} + X_{2} + \dots + X_{n}  - (C_{1}X_{1}^{\prime} + C_{2} * X_{2}^{\prime} + C_{3} * X_{3}^{\prime} + \dots) - 1]^{2} \label{eq:sss}
\end{equation}
In quantum computing, the newly introduced 0-1 variables \( X_{*}^{\prime} \) will increase the number of qubits required for the algorithm. From this, we can see that when we use a quantum computer to solve the DP for a graph \( G = (V,E) \) with \( |V| \) vertices, the number of qubits we need may exceed \( |V| \). According to existing literature, there is no mathematical model for DP that can be modeled with only \( |V| \) qubits. Returning to our previous topic, after we introduce the slack variable \( S \), we next need to represent \( S \) using 0-1 variables, and this representation must encompass the original range of \( S \). Our approach is to use binary numbers to represent \( S \). When \(\{S\} = \{0, 1, 2, \dots, n-1\}\), the specific representation is as shown in Eq. \ref{eq:sc}.
\begin{equation}
  S = \sum\limits_{i=1}^{bl_{n-1}-1} X_{i}^{\prime}*2^{i-1} + (n - 1 - \sum\limits_{i=1}^{bl_{n-1}-1}2^{i-1}) * X_{bl_{n-1}}^{\prime} \label{eq:sc}
\end{equation}

Where \( X_{*}^{\prime} \) is a 0-1 variable. This representation requires the use of \( bl_{n-1} \) additional variables to represent \( S \). This structure has been used by \cite{RN426} to represent flow when solving the maximum flow problem. In this paper, we use it to handle slack variables. Next, we will prove that this representation method allows \( S \) to take any integer in the range \([0, n-1]\). That is, \(\{S\} = \{0,1,2, \dots, n-1\}\).
\begin{theorem}[Theorem 1]
  $\forall n \geq 3$, the range of \( S \) represented in Eq. \ref{eq:sc} encompasses all integers within \([0, n-1]\).
  \label{theorem1}
\end{theorem}

\begin{proof}
    We use mathematical induction to prove this. First, we establish the base case.
  \begin{align*}
    &n=3, S = X_{1}^{\prime} + X_{2}^{\prime}, \{S\} = \{0,1,2\}\\
    &n=4, S = X_{1}^{\prime} + 2 * X_{2}^{\prime}, \{S\} = \{0,1,2,3\}\\
    &n=5, S = X_{1}^{\prime} + 2 * X_{2}^{\prime} + X_{3}^{\prime}, \{S\} = \{0,1,2,3,4\}
  \end{align*}
It can be observed that when \( n = 3, 4, 5 \), the condition clearly holds. \\
Assuming that the statement holds for \( n = k \), that is,\\
\begin{equation*}
    S_{k} = \sum\limits_{i=1}^{bl_{k-1}-1} X_{i}^{\prime}*2^{i-1} + (k - 1 - \sum\limits_{i=1}^{bl_{k-1}-1}2^{i-1}) * X_{bl_{k-1}}^{\prime}, \{S_{k}\} = \{0,1,2,\dots,k-2,k-1\}
\end{equation*}

We need to prove that the statement holds for \( n = k + 1 \).
\begin{equation*}
    S_{k+1} = \sum\limits_{i=1}^{bl_{k}-1} X_{i}^{\prime}*2^{i-1} + (k - \sum\limits_{i=1}^{bl_{k}-1}2^{i-1}) * X_{bl_{k}}^{\prime}, \{S_{k+1}\} = \{0,1,2,\dots,k-1,k\}
\end{equation*}
Proceeding further, we discuss two separate cases.\\
Case one, \( bl_{k-1} \neq bl_{k} \). In this situation, \( k-1 \) must be the largest binary number with \( bl_{k-1} \) digits, and the binary representation of \( k-1 \) must be \( \underbrace{11\ldots11}_{bl_{k-1}} \). Therefore, we can also express \( k-1 \) as:
\begin{equation}
  k-1 = \sum\limits_{i=1}^{bl_{k-1}}2^{i-1}
\end{equation}
Therefore, we have 
\begin{equation}
  k - 1 - \sum\limits_{i=1}^{bl_{k-1}-1}2^{i-1} = 2 ^{bl_{k-1} -1}
\end{equation}
so
\begin{equation}
  S_{k} = \sum\limits_{i=1}^{bl_{k-1}} X_{i}^{\prime}*2^{i-1} 
\end{equation}
as $bl_{k-1} = bl_{k} - 1$ then 
\begin{equation}
  S_{k+1} = S_{k} + (k - \sum\limits_{i=1}^{bl_{k}-1}2^{i-1}) * X_{bl_{k}}^{\prime}
\end{equation}
as $k - \sum\limits_{i=1}^{bl_{k}-1}2^{i-1} = k - \sum\limits_{i=1}^{bl_{k-1}}2^{i-1} =k - (k - 1) = 1$, so $S_{k+1} = S_{k} + X_{bl_{k}} ^{\prime}$. 
Given that \(\{S_{k}\} = \{0,1,2,\dots,k-2,k-1\}\), when \(X_{bl_{k}}^{\prime} = 0\), \(S_{k+1}\) can still take these values, and when \(X_{bl_{k}}^{\prime} = 1\), \(S_{k+1}\) can take the maximum value of \(k - 1 + 1 = k\). Therefore, \(\{S_{k+1}\} = \{0,1,2,\dots,k-1,k\}\).

In the second case, where \( bl_{k-1} = bl_{k} \), we can express \( S_{k} \) and \( S_{k+1} \) as follows:

\begin{align*}
    S_{k} &= X_{1}^{\prime} + 2 * X_{2}^{\prime} + 4 * X_{3}^{\prime} + \dots + 2 ^{m-1} * X_{m}^{\prime} + L * X_{m+1}^{\prime} \\
    S_{k+1} &= X_{1}^{\prime} + 2 * X_{2}^{\prime} + 4 * X_{3}^{\prime} + \dots + 2 ^{m-1} * X_{m}^{\prime} + (L+1) * X_{m+1}^{\prime}
\end{align*}

Here, \( m = bl_{k-1} - 1 = bl_{k} - 1 \), and we let \( Q = \sum_{i=1}^{m-1} 2^{i-1} \). Since the maximum value of \( S_{k} \) is \( k - 1 \), we have \( Q + L = k - 1 \). We divide the expression for \( S_{k+1} \) into two parts: 

\begin{itemize}
    \item Part 1: \( X_{1}^{\prime} + 2 \cdot X_{2}^{\prime} + 4 \cdot X_{3}^{\prime} + \dots + 2^{m-1} \cdot X_{m}^{\prime} \)
    \item Part 2: \( (L+1) \cdot X_{m+1}^{\prime} \)
\end{itemize}

When \( X_{m+1}^{\prime} = 0 \), the value range of Part 1 is \(\{0, 1, 2, \dots, Q - 1, Q\}\). Therefore, we only need to consider whether \( S_{k+1} \) can take the values \(\{Q+1, Q+2, \dots, k - 1, k\}\), which is \(\{Q+1, Q+2, \dots, Q+L, Q+L+1\}\).

It is clear that when \( S_{k+1} \) needs to take values greater than \( Q \), we must set \( X_{m+1}^{\prime} = 1 \), because Part 1 can contribute a maximum value of \( Q \). When Part 2 is involved, in order to achieve \(\{Q+1, Q+2, \dots, Q+L, Q+L+1\}\), Part 1 needs to contribute \(\{Q+1 - (L + 1), Q+2 - (L + 1), \dots, Q+L - (L + 1), Q+L+1 - (L + 1)\}\), which is \(\{Q - L, Q - L + 1, \dots, Q-1, Q\}\).

Since \( k-1 \) is not the largest number with a binary length of \( m+1 \) at this point, \( L \leq 2^{m} - 1 \). Thus, we have:

\[ Q - L = \sum_{i=1}^{m-1} 2^{i-1} - L = 2^{m} - 1 - L \geq 0 \]

This completes the proof.

\end{proof}
    \item Constraint \ref{eq:independent}
    
    For constraint \ref{eq:independent}, we simply need to express it as \( P \cdot (X_{i} \cdot X_{j}) \).
    \item Constraint \ref{eq:perfect}
    
    For constraint \ref{eq:perfect}, we first need to prove the following theorem.
    \begin{theorem}[Theorem 2]
        Within the DP, $\sum\limits_{ij \in E} [X_{i} * (1 - X_{j}) + X_{j} * (1 - X_{i}) ]  \geq |V| - \sum\limits_{i}^{|V|} X_{i}$. \label{theorem:perfect}
    \end{theorem}
      
    \begin{proof}    
      In DP, $|V| - \sum\limits_{i}^{|V|} X_{i}$ denotes the number of vertexes in $V \backslash D$ and $\sum\limits_{ij \in E} [X_{i} * (1 - X_{j}) + X_{j} * (1 - X_{i}) ]$ counts the number of edges with one vertex in $V \backslash D$ and another one in $D$. As we know, each node in $V \backslash D$ is dominated by at least one node in $D$. So $\sum\limits_{ij \in E} [X_{i} * (1 - X_{j}) + X_{j} * (1 - X_{i}) ]$ will not be less than $|V \backslash D|$.        
    \end{proof}
      Leveraging Theorem \ref{theorem:perfect}, we can transform constraint \ref{eq:perfect} into \\
      \begin{equation*}
        P*(\sum\limits_{ij \in E} [X_{i} * (1 - X_{j}) + X_{j} * (1 - X_{i}) ]  - |V| + \sum\limits_{i}^{|V|} X_{i})
      \end{equation*}
      Due to the establishment of the inequality relationship in Theorem \ref{theorem:perfect}, we do not need to square this constraint.

    \item Constraint \ref{eq:clique}

    For constraint \ref{eq:clique}, we need to first prove Theorem \ref{theorem:clique}.
    \begin{theorem}[Theorem 3]
        Within the DP, $\frac{1}{2}c(c - 1) \geq  \sum\limits_{ij \in E} X_{i} * X_{j}$. \label{theorem:clique}
    \end{theorem}
      
      \begin{proof}
      In DP, $c = \sum\limits_{i=1}^{|V|} X_{i}$, which is size of $D$ and $\sum\limits_{ij \in E} X_{i} * X_{j}$ denotes the number of edges in $D$. Apparently, $\frac{1}{2}c(c - 1) \geq  \sum\limits_{ij \in E} X_{i} * X_{j}$ since $\frac{1}{2}c(c - 1)$ denotes the number of edges of a complete graph with $|D|$ nodes. 
      \end{proof}
      Similarly, we can express constraint \ref{eq:clique} as:
      \begin{equation*}
        P*[\frac{1}{2}c(c - 1) -  \sum\limits_{ij \in E} X_{i} * X_{j}]
      \end{equation*}
\end{itemize}

In this section, we have completed the process of transforming the constraints of these problems into quadratic penalties. During this process, we find that for constraints such as constraints \ref{eq:basic}, \ref{eq:total}, and \ref{eq:kdom}, there is a possibility of introducing slack variables. This could lead to a situation where solving these combinatorial optimization problems on a quantum computer requires more qubits than the number of vertices in the $G$. We can estimate the maximum number of qubits required for the methods proposed in this paper.

\begin{theorem}[Theorem 4] \label{theorem:4}
    The upper bound of the number of variables used by the method proposed in this paper to form the DP is $V + 2E$. 
\end{theorem}

\begin{proof}
    For the each node $i$, the size of variables used to represent constraint \ref{eq:basic} is not more than $bl_{d_{i} + 1 - 1}$, and $d_{i}$ is the degree of $i$, since the extra variables is not needed if $d_{i} \leq 2$. So the overall variables is not more than $|V| + \sum\limits_{i=1}^{|V|}bl_{d_{i}}$. And 
    \begin{align*}
        |V| + \sum\limits_{i=1}^{|V|}bl_{d_{i}} \leq|V| + \sum\limits_{i=1}^{|V|} d_{i} = |V|  + 2|E|
    \end{align*}
\end{proof}

According to Theorem \ref{theorem:4}, we need at most \( |V| + 2|E| \) variables to model DP in a graph $G = (V,E)$, \\
and 
\begin{align*}
    &|E|\leq \frac{1}{2}|V|(|V| - 1) \\
    \leftrightarrow \quad & 2|E| \leq |V|^{2} - |V| \\
    \leftrightarrow \quad & |V| + 2|E| \leq |V|^{2}
\end{align*}
This means that compared to the \( |V|^2 \) variables required by \cite{RN416}, the number of qubits we need is fewer.

\subsection{QUBO formulations examples}\label{sec:method4}
For the general QUBO form, such as the example provided in \cite{RN414} in Eq. \ref{eq:exp1}.
\begin{equation}
    minimize \quad y = -5x_{1} - 3x_{2} - 8x_{3} - 6x_{4} + 4x_{1}x_{2} + 8x_{1}x_{3} + 2x_{2}x_{3} + 10x_{3}x_{4} \label{eq:exp1}  
\end{equation}
Since \( x_{i} \in \{0,1\} \), it follows that \( x_{i} = x_{i}^{2} \). We can rewrite the model in its standard form.
\begin{equation} 
    minimize \quad y = 
    \begin{pmatrix}
    
    x_{1} & x_{2} & x_{3} & x_{4}\\
    \end{pmatrix}
    \begin{pmatrix}
    -5 & 2 & 4 & 0\\
    2  & -3& 1 & 0\\
    4  & 1& -8 & 5\\
    0  & 0& 5& -6\\
    \end{pmatrix}
    \begin{pmatrix}
    x_{1}  \\
    x_{2} \\
    x_{3} \\
    x_{4}
    \end{pmatrix}
\end{equation}

In this section, we will provide QUBO models for these variant problems with specific examples for readers to verify. Readers can choose an appropriate penalty coefficient \( P \) based on actual conditions and simplify the following models as described (for example, converting single-variable \( x_{i} \) into \( x_{i}^{2} \) ), ultimately extracting the QUBO matrix. The graph \( G \) we use is shown in Fig. \ref{fig:4node} (this figure is also used by \cite{RN416} to solve the DP on a quantum computer), where the length of each edge is 1. In the following models, \( x_{0}, x_{1}, x_{2}, x_{3} \) are the 0-1 variables associated with the 4 vertices of graph \( G \). For \( x_{i}, i \geq 4 \), these are additional variables introduced to represent slack variables.
\begin{figure}[H]
    \centering
    \includegraphics[width=5cm]{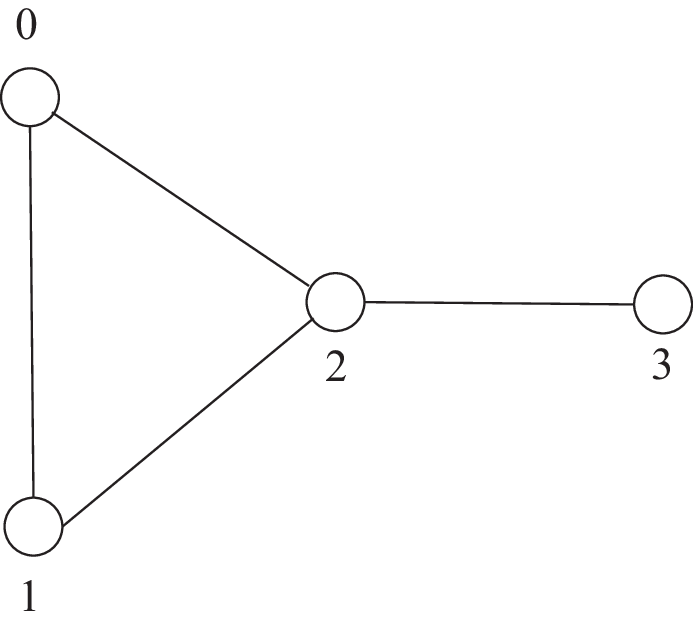}
    \caption {A graph $G$ with 4 nodes.}
    \label{fig:4node}
  \end{figure}
We list the QUBO model corresponding to the DP for readers to verify.
\begin{itemize}
    \item $\gamma$
    \begin{align*}
        \begin{split}
        Minimize \quad f = &x_{0} + x_{1} + x_{2} + x_{3} \\
        &+ P \left(x_{0} + x_{1} + x_{2} - \left(x_{4} + x_{5}\right) - 1\right)^{2} \\
        &+ P \left(x_{1} + x_{0} + x_{2} - \left(x_{6} + x_{7}\right) - 1\right)^{2} \\
        &+ P \left(x_{2} + x_{0} + x_{1} + x_{3} - \left(x_{8} + 2 \cdot x_{9}\right) - 1\right)^{2} \\
        &+ P \left(1 - x_{3} - x_{2} + x_{3} \cdot x_{2}\right)
        \end{split}
    \end{align*}
    
    For the same graph, \cite{RN416} requires 16 qubits, while we only need 10 qubits.
    \item $\gamma_{i}$
    \begin{equation*}
        \begin{split}
        Minimize \quad f = &x_{0} + x_{1} + x_{2} + x_{3} \\
        &+ P \left( x_{0} + x_{1} + x_{2} - (x_{4} + x_{5}) - 1 \right)^{2} \\
        &+ P \left( x_{1} + x_{0} + x_{2} - (x_{6} + x_{7}) - 1 \right)^{2} \\
        &+ P \left( x_{2} + x_{0} + x_{1} + x_{3} - (x_{8} + 2 \cdot x_{9}) - 1 \right)^{2} \\
        &+ P \left( 1 - x_{3} - x_{2} + x_{3} \cdot x_{2} \right) \\
        &+ P \cdot x_{0} \cdot x_{1} + P \cdot x_{0} \cdot x_{2} + P \cdot x_{1} \cdot x_{2} + P \cdot x_{2} \cdot x_{3}
        \end{split}
    \end{equation*}
    \item $\gamma_{t}$
    \begin{equation*}
        \begin{split}
        Minimize \quad f = &\, x_{0} + x_{1} + x_{2} + x_{3} \\
        &+ P \left(1 - x_{1} - x_{2} + x_{1} \cdot x_{2}\right) \\
        &+ P \left(1 - x_{0} - x_{2} + x_{0} \cdot x_{2}\right) \\
        &+ P \left(x_{0} + x_{1} + x_{3} + x_{4} + x_{5} - 1\right)^{2} \\
        &+ P (x_{2} - 1)^{2}
        \end{split}
    \end{equation*}
    \item $\gamma_{per}$
    \begin{equation*}
        \begin{split}
        Minimize \quad f = &\, x_{0} + x_{1} + x_{2} + x_{3} \\
        &+ P \left( x_{0} + x_{1} + x_{2} - (x_{4} + x_{5}) - 1 \right)^{2} \\
        &+ P \left( x_{1} + x_{0} + x_{2} - (x_{6} +  x_{7}) - 1 \right)^{2} \\
        &+ P \left( x_{2} + x_{0} + x_{1} + x_{3} - ( x_{8} + 2 \cdot x_{9}) - 1 \right)^{2} \\
        &+ P \left( 1 - x_{3} - x_{2} + x_{3} \cdot x_{2} \right) \\
        &+ P \left( x_{0} \cdot (1 - x_{1}) + x_{1} \cdot (1 - x_{0}) + x_{0} \cdot (1 - x_{2}) + x_{2} \cdot (1 - x_{0}) \right. \\
        &\quad \left. + x_{1} \cdot (1 - x_{2}) + x_{2} \cdot (1 - x_{1}) + x_{2} \cdot (1 - x_{3}) + x_{3} \cdot (1 - x_{2}) \right) \\
        &- P \cdot 4 + P \cdot (x_{0} + x_{1} + x_{2} + x_{3})
        \end{split}
    \end{equation*}
    \item $\gamma_{cl}$
    \begin{equation*}
        \begin{split}
        Minimize \quad f = &\, x_{0} + x_{1} + x_{2} + x_{3} \\
        &+ P \left( x_{0} + x_{1} + x_{2} - ( x_{4} +  x_{5}) - 1 \right)^{2} \\
        &+ P \left( x_{1} + x_{0} + x_{2} - ( x_{6} +  x_{7}) - 1 \right)^{2} \\
        &+ P \left( x_{2} + x_{0} + x_{1} + x_{3} - ( x_{8} + 2 \cdot x_{9}) - 1 \right)^{2} \\
        &+ P \left( 1 - x_{3} - x_{2} + x_{3} \cdot x_{2} \right) \\
        &+ P \left( 0.5 \cdot (x_{0} + x_{1} + x_{2} + x_{3}) \cdot (x_{0} + x_{1} + x_{2} + x_{3} - 1) \right. \\
        &\quad \left. - (x_{0} \cdot x_{1} + x_{0} \cdot x_{2} + x_{1} \cdot x_{2} + x_{2} \cdot x_{3}) \right)
        \end{split}
    \end{equation*}
    \item $\gamma_{iper}$
    \begin{equation*}
        \begin{split}
        Minimize \quad f = &\, x_{0} + x_{1} + x_{2} + x_{3} \\
        &+ P \left( x_{0} + x_{1} + x_{2} - (x_{4} + x_{5}) - 1 \right)^{2} \\
        &+ P \left( x_{1} + x_{0} + x_{2} - (x_{6} + x_{7}) - 1 \right)^{2} \\
        &+ P \left( x_{2} + x_{0} + x_{1} + x_{3} - (x_{8} + 2 \cdot x_{9}) - 1 \right)^{2} \\
        &+ P \left( 1 - x_{3} - x_{2} + x_{3} \cdot x_{2} \right) \\
        &+ P \cdot x_{0} \cdot x_{1} + P \cdot x_{0} \cdot x_{2} + P \cdot x_{1} \cdot x_{2} + P \cdot x_{2} \cdot x_{3} \\
        &+ P \left( x_{0} \cdot (1 - x_{1}) + x_{1} \cdot (1 - x_{0}) + x_{0} \cdot (1 - x_{2}) + x_{2} \cdot (1 - x_{0}) \right. \\
        &\quad \left. + x_{1} \cdot (1 - x_{2}) + x_{2} \cdot (1 - x_{1}) + x_{2} \cdot (1 - x_{3}) + x_{3} \cdot (1 - x_{2}) \right) \\
        &- P\cdot4 + P \cdot (x_{0} + x_{1} + x_{2} + x_{3})
        \end{split}
    \end{equation*}
    \item $\gamma_{tper}$
    \begin{equation*}
        \begin{split}
        Minimize \quad f = &\, x_{0} + x_{1} + x_{2} + x_{3} \\
        &+ P \left(1 - x_{1} - x_{2} + x_{1} \cdot x_{2}\right) \\
        &+ P \left(1 - x_{0} - x_{2} + x_{0} \cdot x_{2}\right) \\
        &+ P \left(x_{0} + x_{1} + x_{3} +  x_{4} +  x_{5} - 1\right)^{2} \\
        &+ P \left((x_{2} - 1)^{2}\right) \\
        &+ P \left( x_{0} \cdot (1 - x_{1}) + x_{1} \cdot (1 - x_{0}) + x_{0} \cdot (1 - x_{2}) + x_{2} \cdot (1 - x_{0}) \right. \\
        &\quad \left. + x_{1} \cdot (1 - x_{2}) + x_{2} \cdot (1 - x_{1}) + x_{2} \cdot (1 - x_{3}) + x_{3} \cdot (1 - x_{2}) \right) \\
        &- P\cdot4 + P \cdot (x_{0} + x_{1} + x_{2} + x_{3})
        \end{split}
    \end{equation*}
    \item $\gamma_{k}, \quad k = 2$
    \begin{equation*}
        \begin{split}
        Minimize \quad f = &\, x_{0} + x_{1} + x_{2} + x_{3} \\
        &+ P \left(x_{0} + x_{1} + x_{2} + x_{3} - (x_{4} + 2 \cdot x_{5}) - 1\right)^{2} \\
        &+ P \left(x_{1} + x_{0} + x_{2} + x_{3} - (x_{6} + 2 \cdot x_{7}) - 1\right)^{2} \\
        &+ P \left(x_{2} + x_{0} + x_{1} + x_{3} - (x_{8} + 2 \cdot x_{9}) - 1\right)^{2} \\
        &+ P \left(x_{3} + x_{2} + x_{0} + x_{1} - (x_{10} + 2 \cdot x_{11}) - 1\right)^{2}
        \end{split}
    \end{equation*}
\end{itemize}

\section{Conclusion}\label{sec:conclusion}

This paper primarily investigates the QUBO modeling of the DP and its variants. We first modeled the DP and its seven variant problems as 0-1 integer programming models. Then, we analyzed the constraints of these problems and categorized them into four types: (\ref{eq:basic}, \ref{eq:total}, \ref{eq:kdom}), \ref{eq:independent}, \ref{eq:perfect}, and \ref{eq:clique}. For each type of constraint condition, we provided methods to convert them into quadratic penalties. By adding the obtained quadratic penalties to the original objective function, we completed the QUBO modeling for the DP problem and its variants. When dealing with constraints (\ref{eq:basic}, \ref{eq:total}, \ref{eq:kdom}), we used a transformation method different from \cite{RN416}, which allows us to model DP with \( |V| + 2|E| \) variables. This will lower the barrier for solving DP using quantum computers. The research findings on QUBO modeling for DP and its variant problems in this paper will enable the solution of such problems on quantum computers, greatly accelerating the development of quantum algorithms for these problems, especially the variants. Future research mainly includes using quantum algorithms such as QAOA and QA to solve these QUBO models or finding more suitable QUBO modeling methods, such as those with fewer variables.

\label{}

\section*{Data Availability}
The data used to support the findings of this study are included within the article.

\section*{Conflicts of interest}
The authors declare that they have no conflicts of interest that could have appeared to influence the work reported in this paper.

\section*{Funding statement}
 
This work is supported by National Natural Science Foundation of China (No. 12331014).

\section*{Declaration of Generative AI and AI-assisted technologies in the writing process}

During the preparation of this work the authors used chatgpt in order to improve readability and language. After using this tool, the authors reviewed and edited the content as needed and take full responsibility for the content of the publication.



 \bibliographystyle{elsarticle-harv} 
 \bibliography{qubo}






\end{CJK}
\end{document}